\documentclass[times, 10pt,twocolumn]{article}
\usepackage{latex8}
\usepackage{times}

\usepackage{graphicx}
\usepackage{amsfonts}
\usepackage{amssymb}
\usepackage{amsmath}

\usepackage{algorithm}
\usepackage{algpseudocode}
\usepackage{algorithmicx}

\newcommand{\comment}[1]{}

\newtheorem{theorem}{Theorem}[section]
     \newtheorem{lemma}[theorem]{Lemma}
     
     \newtheorem{corollary}[theorem]{Corollary}

     \newenvironment{definition}[1][Definition:]{\begin{trivlist}
     \item[\hskip \labelsep {\bfseries #1}]}{\end{trivlist}}

     \newcommand{\qed}{\nobreak \ifvmode \relax \else
           \ifdim\lastskip<1.5em \hskip-\lastskip
           \hskip1.5em plus0em minus0.5em \fi \nobreak
           \vrule height0.75em width0.5em depth0.25em\fi}

\newcounter{remark_ordering}
\newtheorem{remark}[remark_ordering]{Remark}

\comment{

}

     \newenvironment{proof}[1][Proof]{\begin{trivlist}
     \item[\hskip \labelsep {\bfseries #1}]}{\end{trivlist}}

\begin{document}

\title{  Novel Blind Signal Classification Method Based on Data Compression }

\author{Xudong Ma\\
FlexRealm Silicon Inc., Virginia, U.S.A. \\
 xma@ieee.org\\
}

\comment{
\author{\authorblockN{Xudong Ma \\}
\authorblockA{FlexRealm Silicon Inc., Virginia, U.S.A.\\
Email: xma@ieee.org} }}

\maketitle \thispagestyle{empty}

\begin{abstract}

This paper proposes a novel algorithm for signal classification
problems. We consider a non-stationary random signal, where samples
can be classified into several different classes, and samples in
each class are identically independently distributed with an unknown
probability distribution. The problem to be solved is to estimate
the probability distributions of the classes and the correct
membership of the samples to the classes. We propose a signal
classification method based on the data compression principle that
the accurate estimation in the classification problems induces the
optimal signal models for data compression. The method formulates
the classification problem as an optimization problem, where a so
called { ``classification gain''} is maximized. In order to
circumvent the difficulties in integer optimization, we propose a
continuous relaxation based algorithm. It is proven in this paper
that asymptotically vanishing optimality loss is incurred by the
continuous relaxation. We show by simulation results that the
proposed algorithm is effective, robust and has low computational
complexity. The proposed algorithm  can be applied to solve various
multimedia signal segmentation, analysis, and pattern recognition
problems.

\comment{ 6th International Conference on

Information Technology : New Generations

ITNG 2009

April 27-29, 2009, Las Vegas, Nevada, USA

www.itng.info}

\end{abstract}

\comment{
}

\section{Introduction}

\label{section_introduction}

Nature multimedia signals are non-stationary in nature. For example,
the statistical property of a nature image can vary significantly
across edges; an  audio signal may contains silent segments and
active segments; and the statistics of a video signal can be totally
different before and after a change-of-scene.

Therefore, it is not a surprise that signal classification problems
arise naturedly in many scenarios of multimedia signal processing.
That is, the signal samples need to be classified into different
classes, where each class contains only signals with homogeneous
statistics. Such signal classification problems have been
extensively discussed under the name of {\it thresholding } or {\it
segmentation} (see for instance \cite{sezgin04},
\cite{srinvasan05}). The applications  range from multimedia signal
enhancement to multimedia content analysis and understanding.

In this paper, we propose an original signal classification method
based on data compression principles. The center idea of our
approach is that signal classifications can be considered as
operations of signal modeling. If a mismatched signal model is used
in  data compression, then a performance loss in terms of coding
efficiency is incurred. Therefore, an accurate signal classification
result should maximize the coding efficiency in data compression.

Based on the above data compression principles, we propose an
optimization formulation of the signal classification problems. In
the optimization formulation, the optimization variables are the
memberships  of samples to different classes, and the objective
function is the coding efficiency. More precisely, we optimize the
{\it classification gain}, which is a measure of coding efficiency.
In order to avoid the difficulties in discrete optimizations, we
further propose a continuous relaxation and random rounding solution
for the optimization problems. It is proven in this paper that the
optimality loss due to relaxation vanishes when the total sample
number is large.

In the data compression literature, the adaptive coding approach
based on classification has been previously discussed. Early works
on classifying DCT and wavelet coefficients into classes and using
individual quantizer for each class include \cite{chen77},
\cite{woods86}, \cite{tse92}. The term {\it classification gain} is
coined by Joshi, Jafarkhani, Kasner, Fischer, Farvardin, Marcelin,
and Bamberger \cite{joshi97}. Two signal classification algorithms,
the maximum classification gain and equal mean-normalized standard
deviation classification, have been proposed in \cite{joshi97}. The
signal classification approaches for adaptive coding have also been
adopted in state-of-the-art subband coding schemes (see for instance
\cite{yoo99}).

The signal classification problem can also be considered as an
unsupervised pattern recognition problem. In the pattern recognition
literature, clustering algorithms for such problems have been
previously discussed. The well-known algorithms include the K-means
algorithms, and Expectation-Maximization (EM) algorithms
\cite{duda00}, \cite{theodoridis06}. In the K-means algorithms, the
classification problem is formulated as an optimization problem,
where a sum of distances is minimized by an iterative approach. The
classification problem can also be formulated as an estimation with
incomplete data problem, and solved by the EM algorithms
\cite{dempster77}. In the EM algorithms, the log likelihoods of the
estimated distribution parameters are iteratively lower bounded and
maximized.

There are several advantages of the proposed algorithm over the
previous algorithms. First, the proposed algorithm is ``provably
good'', i.e., the algorithm is amendable to rigorous theoretical
analysis. Second, the proposed algorithm is more tractable due to
that difficult integer optimizations are avoided. The proposed
approach is also a more general and flexible framework. For example,
the proposed approach is more flexible in choosing optimization
solvers. The proposed algorithm can achieve global optimal solutions
if a global optimization solver is used; while both K-means
algorithms and EM algorithms converge to local optimal solutions.

\comment{ The difference between the proposed classification
algorithm in this paper and the previously existing algorithm is
that the proposed approach is based on continuous relaxations. The
proposed algorithm is more tractable due to that difficult discrete
optimization is avoided. The proposed algorithm is also more
amenable to rigorous theoretical discussions. Another difference
between the approach in this paper to the previous approaches is
that all signal are assumed to have zero mean in the previous
approaches.}

\emph{Our contribution:} In summary, we propose a novel principle of
signal classification  based on data compression. We propose the
continuous relaxation solution for the optimization formulation. We
prove that the optimality loss due to continuous relaxation vanishes
asymptotically with respect to the sample number.

\emph{Organization of the paper:} The rest of this paper is
organized as follows. We discuss the signal model in Section
\ref{section_system}. A review of classification gain is provided in
Section \ref{section_class_gain}. We present the proposed
classification algorithm in Section
\ref{section_classification_algorithm}. We present a theoretical
discussion on the optimality loss due to continuous relaxation in
Section \ref{sec_performance_analysis}. Numerical results are
presented in Section \ref{sec_numerical}. We present the conclusion
remark in Section \ref{sec_conclusion}.

\section{Signal Model}
\label{section_system}

\comment{
\begin{figure}[h]
 \centering
 \includegraphics[width=3in]{signal_class_model.eps}
 \caption{The signal model}
 \label{fig_signal_model}
\end{figure}}

In this paper, we consider a random signal with $N$ samples, $
x_1,x_2,\ldots,x_N$. We assume that the random signal is
non-stationary and is a mixture of samples from $J$ memoryless
information sources. That is, there exist $J$ memoryless information
sources. The corresponding probability distribution for the $i$-th
information source is $P_i$. For each n, $1\leq n\leq N$, the random
variable $x_n$ is distributed with one of the distributions $P_i$
and independent of all the other signal samples.

\comment{A block diagram of the signal model is shown in Fig.
\ref{fig_signal_model}.}

The considered signal classification problem is thus to estimate the
probability distribution of each information source and the
membership of each signal sample. We assume that the probability
distributions of all information sources are unknown, i.e., we
consider a blind signal classification scenario. The case, where the
probability distribution of each information source is known,  can
be straightforwardly solved by using the first principles of
statistical detection and estimation theory, and thus will not be
discussed. In this paper, we also assume that all information
sources are Gaussian distributed. The generalization of the proposed
algorithm to non-Gaussian cases (and also an information theoretic
analysis of the algorithm) will be discussed in a companion paper
\cite{ma08}.

\comment{ In this paper, we assume that the probability
distributions of all information sources are Gaussian. The cases of
non Gaussian will be discussed in a companion paper \cite{ma08}.

In the rest part of this paper, we assume that all signals are
mixture of samples from Gaussian information sources. That is, each
signal sample is generated from one of $J$ information sources and
Gaussian distributed.  The considered signal classification problem
is to estimate the probability distributions of each information
source and the membership of each signal sample. We assume that the
probability distributions of all information sources are unknown,
i.e., we consider a blind signal classification scenario. The case,
where the probability distribution of each information source is
known,  can be straightforwardly solved by using the first
principles of statistical detection and estimation theory, and
therefore will not be discussed. The generalization of the proposed
algorithm to non-Gaussian cases (and also an information theoretic
analysis of the algorithm) will be discussed in a companion paper
\cite{ma08}.}

\section{Classification Gain}
\label{section_class_gain}

According to the rate-distortion theory \cite{cover}, for a
memoryless Gaussian information source with variance $\sigma^2$, if
an encoder with rate $R$ is used, then the smallest achievable
mean-squared  error distortion is,
\begin{align}
D(R)=\sigma^22^{-2R}. \label{rd_fun_gaussian}
\end{align}
The function $D(R)$ is the distortion-rate function for Gaussian
information sources. For non-Gaussian information sources with the
same variance $\sigma^2$, the distortion $D(R)$ is  achievable  by
using a source encoder designed for Gaussian sources \cite{ma08}.

For the non-stationary random signal $x_1,x_2,\ldots,x_N$, there are
two approaches to encode the signal. A naive approach adopts an
encoder designed for Gaussian sources to encode all signal samples.
The achievable distortion is
\begin{align}
\sigma_x^22^{-2R},
\end{align}
where, $\sigma_x^2$ is the variance for the random signal
$x_1,\ldots,x_N$. A better approach first classifies the signal
samples into $J$ different classes, and then uses an individual
encoder for each class of samples. Denote the number of samples in
the $i$-th class by $N_i$. Define $p_i$ as the fraction $p_i=N_i/N$.
Denote the variance of samples in the $i$-th class by $\sigma_i^2$.
Under an arbitrary rate allocation, the achievable expected
distortion is,
\begin{align}
\sum_{i=1}^{J}p_i\sigma_i^22^{-2R_i},
\end{align}
where  $R_i$ is the rate allocated to encode the samples in the
$i$-th class,
\begin{align}
\sum_{i=1}^{J}p_iR_i=R-H(p_1,p_2,\ldots,p_J),
\end{align}
and $H(p_1,p_2,\ldots,p_J)$ is the entropy function with base $2$.
It can be easily found by using the Lagrangian multiple method, that
the optimal rate allocation satisfies the following condition
\begin{align}
R_i=\max\left\{\frac{1}{2}\log_2\left(\frac{\sigma_i^2}{\lambda}\right),0\right\}.
\end{align}
Assume that the rate $R$ is sufficiently high, so that $R_i>0$ for
all i, $1\leq i\leq J$. Then, the optimal achievable distortion is,
\begin{align}
\left(\prod_{i=1}^{J}\left(\sigma_i^2\right)^{p_i}\right)2^{-2R+2H(p_1,\ldots,p_J)}.
\end{align}
As in the previous research, we define the classification gain as
the ratio of two achievable distortions,
\begin{align}
G=\frac{\sigma_x^2}{
2^{2H(p_1,\ldots,p_J)}\prod_{i=1}^{J}\left(\sigma_i^2\right)^{p_i}
}.
\end{align}

\section{Classification Algorithm}
\label{section_classification_algorithm}

In this section, we present the proposed signal classification
algorithm. The algorithm is based on the principle that the optimal
classification induces the optimal signal model for data compression
(a rigorous treatment of this argument can be found in \cite{ma08}).
We formulate the classification problem as an integer optimization
where the classification gain is maximized.

The integer optimization is as follows.
\begin{align}
& \mbox{(Integer)}\,\,\,  \min
\left(\prod_{i=1}^{J}\left(\sigma_i^2\right)^{p_i}\right)2^{2H(p_1,\ldots,p_J)} \\
& \mbox{Subject to:} \nonumber \\
& \mu_i=\frac{\sum_{n=1}^{N}a_{ni}x_n}{\sum_{n=1}^{N}a_{ni}} \\
& \sigma_i^2=
\left(\frac{1}{\sum_{n=1}^{N}a_{ni}}\right)\sum_{n=1}^{N}a_{ni}\left(x_n-
\mu_i\right)^2\\
& p_i=\frac{\sum_{n=1}^{N}a_{ni}}{N} \\
& \sum_{i=1}^{J} a_{ni}=1, \mbox{ for any }n, 1\leq n\leq N \\
& a_{ni}\in \{0,1\}
\end{align}
In the integer optimization, the optimization variables are
variables $a_{ni}$, $1\leq n \leq N$, $1\leq i \leq J$. Each
variable $a_{ni}$ is a binary variable indicating the membership of
the $n$th sample, i.e.,
\begin{align}
a_{ni}=\left\{\begin{array}{ll} 1, & \mbox{if the }n\mbox{th sample
is classified to the }i\mbox{th class} \\
0, & \mbox{otherwise}
\end{array}\right.
\end{align}
Alternatively, we can also use a set of integers
$z_1,z_2,\ldots,z_N$ to represent the membership of the signal
samples. The integer $z_n=i$, if and only if the $n$th signal sample
is classified to the $i$th class. In the sequel, we will call such a
set of integers $z_1,\ldots,z_n$ a {\it classification scheme}.

Because integer programming is generally difficult to solve, we
propose a relaxation and random rounding approach. The relaxed
programming is as follows.
\begin{align}
& \mbox{(Relaxation)}\,\,\,  \min
\left(\prod_{i=1}^{J}\left(\sigma_i^2\right)^{p_i}\right)2^{2H(p_1,\ldots,p_J)} \\
& \mbox{Subject to:} \nonumber \\
& \mu_i=\frac{\sum_{n=1}^{N}a_{ni}x_n}{\sum_{n=1}^{N}a_{ni}} \\
&
\sigma_i^2= \left(\frac{1}{\sum_{n=1}^{N}a_{ni}}\right)\sum_{n=1}^{N}a_{ni}\left(x_n-\mu_i\right)^2\\
& p_i=\frac{\sum_{n=1}^{N}a_{ni}}{N} \\
& \sum_{i=1}^{J} a_{ni}=1, \mbox{ for any }n, 1\leq n\leq N \\
& 0\leq a_{ni}\leq 1
\end{align}
In the relaxed programming, the 0-1 constraints have been relaxed to
box constraints.

The proposed algorithm is summarized in Algorithm
\ref{relaxation_algorithm}. In the first step, the relaxed
optimization is solved. Denote the solution of the relaxed
optimization by $a_{ni}^\ast$. In the random rounding step, we
randomly set $z_n$ according to the values of $a_{ni}^{\ast}$. That
is, $ {\mathbb P}(z_n=i)=a_{ni}^{\ast}$.

\begin{algorithm}
\caption{The blind signal classification algorithm}
\label{relaxation_algorithm}
\begin{algorithmic}
\Procedure{blind classification}{$x_1,x_2,\ldots,x_N,J$}

\State solve the relaxed optimization problem

\For{$n\gets 1,N$}

\State randomly set $z_n=i$ with probability $a_{ni}^{\ast}$

\EndFor

\State \textbf{Return} {classification scheme $z_1,z_2,\ldots,z_N$}
\EndProcedure
\end{algorithmic}
\end{algorithm}

\comment{
\begin{algorithm}
\caption{The blind signal classification algorithm}
\label{relaxation_algorithm} \begin{algorithmic}\STATE{input signal
$x_1,x_2,\ldots,x_N$} \STATE{Classification $z_1,z_2,\ldots,z_N$}
\end{algorithmic}
\end{algorithm}
}

\section{Performance Analysis}
\label{sec_performance_analysis}

In this section, we  present a performance analysis of the proposed
classification algorithm. We  show that the optimality loss due to
relaxation and random rounding is negligible if the total sample
number $N$ is sufficiently large. Therefore, our algorithm is
near-optimal with reduced computational complexity.

We need to use the inequality in Lemma \ref{azuma_inequality} in our
discussion. The inequality is one variation of the Azuma inequality
proven by Janson \cite{azuma67}\cite{janson98}.

\begin{lemma} \cite{janson98}
\label{azuma_inequality}\emph{(Azuma Inequality)} Let
$Z_1,\dots,Z_N$ be independent random variables, with $Z_k$ taking
values in a set $\Lambda_k$. Assume that a (measurable) function
$f:\Lambda_1\times \Lambda_2\times \cdots \times
\Lambda_N\rightarrow {\mathbb R}$ satisfies the following Lipschitz
condition (L).
\begin{itemize}
\item (L) If the vectors $z,z'\in\prod_{1}^{N}\Lambda_i$ differ only
in the $k$th coordinate, then $|f(z)-f(z')|<c_k$, $k=1,\ldots,N$.
\end{itemize}
Then, the random variable $X=f(Z_1,\ldots,Z_N)$ satisfies, for any
$t\geq 0$,
\begin{align}
{\mathbb P}(X\geq {\mathbb E}X+t)\leq
\exp\left(\frac{-2t^2}{\sum_{1}^{N}c_k^2}\right),
\end{align}
\begin{align}
{\mathbb P}(X\leq {\mathbb E}X-t)\leq
\exp\left(\frac{-2t^2}{\sum_{1}^{N}c_k^2}\right).
\end{align}
\end{lemma}

As in the previous sections, we use $a_{ni}^{\ast}$ to denote the
solution for the relaxation programming. We use  $p_i^{\ast}$,
$\left(\sigma_i^{\ast}\right)^2$, $\mu_i^{\ast}$ to denote the
corresponding occurrence probability, variance and mean. That is,
\begin{align}
& \mu_i^{\ast}=\frac{\sum_{n=1}^{N}a_{ni}^{\ast}x_n}{\sum_{n=1}^{N}a_{ni}^{\ast}}, \\
& (\sigma_i^\ast)^2= \left(\frac{1}{\sum_{n=1}^{N}a_{ni}^\ast}\right)\sum_{n=1}^{N}a_{ni}^\ast\left(x_n-\mu_i^\ast\right)^2,\\
& p_i^\ast=\frac{\sum_{n=1}^{N}a_{ni}^\ast}{N}.
\end{align}
We  use $z_1,\ldots,z_N$ to denote the classification scheme
obtained from Algorithm \ref{relaxation_algorithm}. In the
following, we abuse the notation and use $a_{ni}$ to denote the
randomly rounded version of the variable $a_{ni}^\ast$, i.e.,
\begin{align}
a_{ni}=\left\{\begin{array}{ll}
1, & \mbox{if }z_n=i \\
0, & \mbox{otherwise}
\end{array}\right.
\end{align}
Similarly, we use  $p_i$, $\sigma_i^2$, $\mu_i$ to denote the
corresponding occurrence probability, variance, and mean. That is,
\begin{align}
& \mu_i=\frac{\sum_{n=1}^{N}a_{ni}x_n}{\sum_{n=1}^{N}a_{ni}}, \\
& \sigma_i^2= \left(\frac{1}{\sum_{n=1}^{N}a_{ni}}\right)\sum_{n=1}^{N}a_{ni}\left(x_n-\mu_i\right)^2,\\
& p_i=\frac{\sum_{n=1}^{N}a_{ni}}{N}.
\end{align}

\begin{definition}
Let $\epsilon_1,\epsilon_2,\epsilon_3$ be arbitrary positive real
numbers. We say that one classification scheme is
$(\epsilon_1,\epsilon_2,\epsilon_3)$-typical if the following
conditions hold for all $i$, $1\leq i\leq J$,
\begin{align}
\left|\sum_{n=1}^{N}a_{ni}-\sum_{n=1}^{N}a_{ni}^\ast\right|\leq
\epsilon_1 N,
\end{align}
\begin{align}
\left|\sum_{n=1}^{N}a_{ni}x_n-\sum_{n=1}^{N}a_{ni}^\ast
x_n\right|\leq \epsilon_2 N,
\end{align}
\begin{align}
\left|\sum_{n=1}^{N}a_{ni}\left(x_n-\mu_i^\ast\right)^2-\sum_{n=1}^{N}a_{ni}^\ast\left(x-\mu_i^\ast\right)^2\right|\leq
\epsilon_3 N.
\end{align}
\end{definition}

\begin{lemma}
If $\epsilon_1,\epsilon_2,\epsilon_3$ all go to zero, then for
$(\epsilon_1,\epsilon_2,\epsilon_3)$-typical classification schemes,
$\mu_i$, $p_i$, $\sigma_i^2$ go to  $\mu_i^\ast$, $p_i^\ast$,
$\left(\sigma_i^\ast\right)^2$ respectively.
\end{lemma}
\begin{proof}
It can be easily checked that $\mu_i$ goes to $\mu_i^\ast$, and
$p_i$ goes to $p_i^\ast$. For $\sigma_i^2$, we notice that
\begin{align}
& \sum_{n=1}^{N}a_{ni}(x_n-\mu_i)^2 \\
& =\sum_{n=1}^{N}a_{ni}(x_n-\mu_i^{\ast}+\mu_i^{\ast}-\mu_i)^2\\
& =\sum_{n=1}^{N}a_{ni}(x_n-\mu_i^{\ast})^2+
\sum_{n=1}^{N}a_{ni}(\mu_i^{\ast}-\mu_i)^2 \\
& \hspace{0.5in} +2
\sum_{n=1}^{N}a_{ni}(x_n-\mu_i^{\ast})(\mu_i^{\ast}-\mu_i) \\
& =\sum_{n=1}^{N}a_{ni}(x_n-\mu_i^{\ast})^2+
p_iN(\mu_i^{\ast}-\mu_i)^2 \\
& \hspace{0.5in} +2
(\mu_i^{\ast}-\mu_i)\sum_{n=1}^{N}a_{ni}(x_n-\mu_i^{\ast}) \\
& =\sum_{n=1}^{N}a_{ni}(x_n-\mu_i^{\ast})^2-
p_iN(\mu_i^{\ast}-\mu_i)^2
\end{align}
Therefore,
\begin{align}
\sigma_i^2 & =\frac{\sum_{n=1}^{N}a_{ni}(x_n-\mu_i)^2}{\sum_{n=1}^{N}a_{ni}}\\
&
=\frac{\sum_{n=1}^{N}a_{ni}(x_n-\mu_i^{\ast})^2}{\sum_{n=1}^{N}a_{ni}}-(\mu_i^{\ast}-\mu_i)^2
\end{align}
It follows that $\sigma_i^2$ goes to
$\left(\sigma_i^\ast\right)^2$.\qed
\end{proof}

\begin{theorem}
\label{main_theorem} Let $\epsilon_1,\epsilon_2,\epsilon_3$ be
arbitrary positive real numbers. Let $V=\max_n x_n-\min_n x_n$.
Then, the probability that the classification scheme obtained from
Algorithm \ref{relaxation_algorithm} is not
$(\epsilon_1,\epsilon_2,\epsilon_3)$-typical is upper bounded as
follows.
\begin{align}
& {\mathbb P}\left(\mbox{the classification scheme is not
}(\epsilon_1,\epsilon_2,\epsilon_3)\mbox{-typical}\right) \\
& \leq 2J
\exp\left(-2\epsilon_1^2N\right)+2J\exp\left(\frac{-2\epsilon_2^2N}{V^2}\right)+
2J\exp\left(\frac{-2\epsilon_3^2N}{V^4}\right)
\end{align}
\end{theorem}
\begin{proof}
By using the Azuma inequality, we can show that
\begin{align}
{\mathbb
P}\left(\left|\sum_{n=1}^{N}a_{ni}-\sum_{n=1}^{N}a_{ni}^\ast\right|\geq
\epsilon_1 N\right)\leq 2\exp\left(-2\epsilon_1^2N\right),
\end{align}
\begin{align}
{\mathbb
P}\left(\left|\sum_{n=1}^{N}a_{ni}x_n-\sum_{n=1}^{N}a_{ni}^\ast
x_n\right|\geq \epsilon_2 N\right)\leq
2\exp\left(\frac{-2\epsilon_2^2N}{V^2}\right),
\end{align}
\begin{align}
& {\mathbb
P}\left(\left|\sum_{n=1}^{N}a_{ni}\left(x_n-\mu_i^\ast\right)^2-\sum_{n=1}^{N}a_{ni}^\ast\left(x-\mu_i^\ast\right)^2\right|\geq
\epsilon_3 N\right) \\
& \leq 2\exp\left(\frac{-2\epsilon_3^2N}{V^4}\right).
\end{align}
The theorem follows from a union bound.\qed
\end{proof}

\begin{corollary}
\label{main_corollary}  If the sample number $N$ is sufficiently
large, then the classification scheme obtained from Algorithm
\ref{relaxation_algorithm} is
$(\epsilon_1,\epsilon_2,\epsilon_3)$-typical with probability close
to one.
\end{corollary}
\begin{proof}
The upper bound in Theorem \ref{main_theorem} is close to zero for
sufficiently large $N$.\qed
\end{proof}

\begin{corollary}
\label{existence_corollary}  If the sample number $N$ is
sufficiently large, then there exists at least one
$(\epsilon_1,\epsilon_2,\epsilon_3)$-typical classification scheme.
\end{corollary}
\begin{proof}
We have presented an algorithm, which constructs such a
classification scheme with success probability close to one.\qed
\end{proof}

\begin{remark}
Theorem \ref{main_theorem} and Corollary \ref{existence_corollary}
imply that the gap between the optimal classification gain achieved
in the relaxation optimization and the optimal classification gain
achieved in the integer optimization goes to zero asymptotically. In
other words, the continuous relaxation incurs an asymptotically
vanishing optimality loss.
\end{remark}


%

\section{Numerical Results }

\label{sec_numerical}

In this section, we present numerical results for the proposed blind
classification algorithm.  The classification error probabilities
are measured by false classification ratios, $r_c=m_i/n_i$, where
$n_i$ denotes the number of samples belong to the class $i$, and
$m_i$ denotes the number of samples belong to the class $i$ and are
classified to classes other than the class $i$. The IPOPT package is
used to solve the optimization programming \cite{wachter06}.

In Fig. \ref{case_one}, we depict the result of the proposed
algorithm for a one dimensional mixed signal of two classes, with
one class having mean 128 and variance 16, and the other class
having mean 16 and variance 16. The false classification ratios are
all $0$\%. In Fig. \ref{case_two}, we depict the result of the
proposed algorithm for a one dimensional mixed signal of two
classes, with one class having mean 128 and variance 2500, and the
other class having mean 128 and variance 25. The false
classification ratios are $16.41$\% and $6.25$\%. In Fig.
\ref{case_three}, we depict the result of the proposed algorithm for
a one dimensional mixed signal of two classes, with one class having
mean 50 and variance 2500, and the other class having mean 5 and
variance 25. The false classification ratios are $10.16$\% and
$3.91$\%. In each figure, the signal is shown in the upper part of
the figure. The classification result is shown in the lower part of
the figure. The grey region of the bar indicates the samples which
are classified into one class, and the white region of the bar
indicates the samples which are classified into the other class. In
all the three cases, the signal sample number $N=256$.

\begin{figure}[h]
 \centering
 \includegraphics[width=3in]{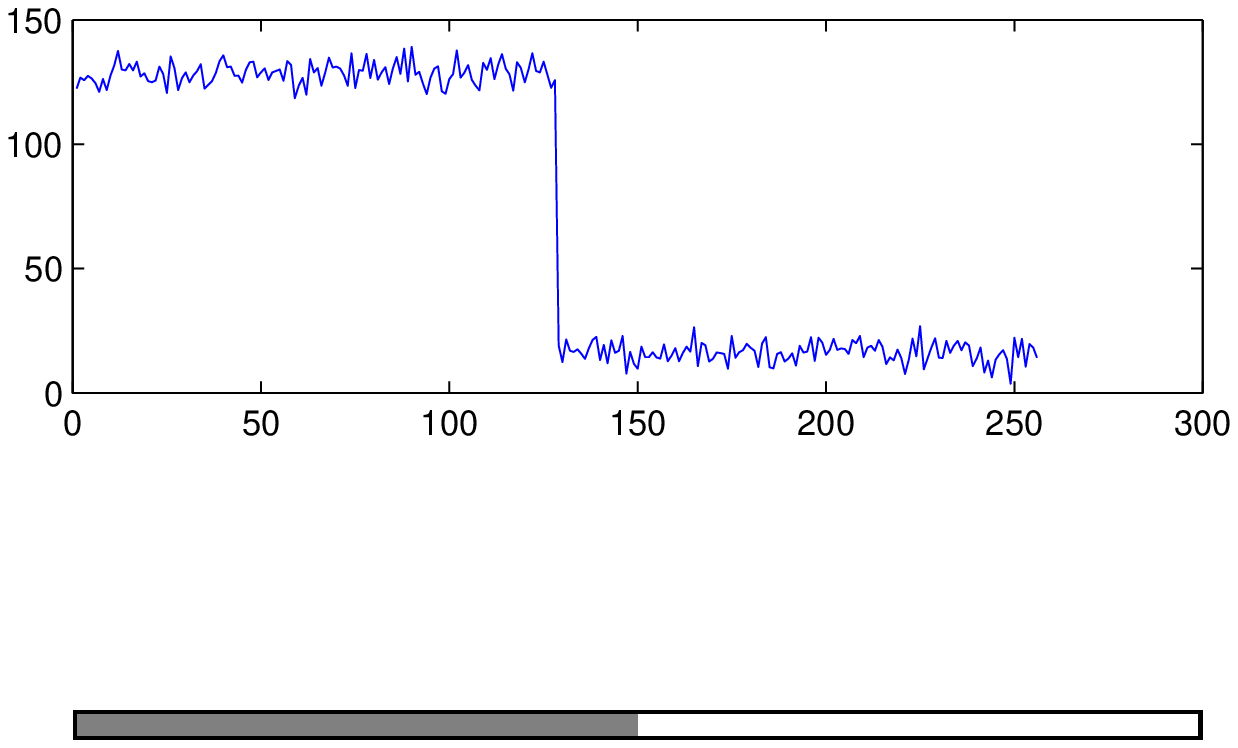}
 \caption{One dimensional, two classes.}
 \label{case_one}
\end{figure}

\begin{figure}[h]
 \centering
 \includegraphics[width=3in]{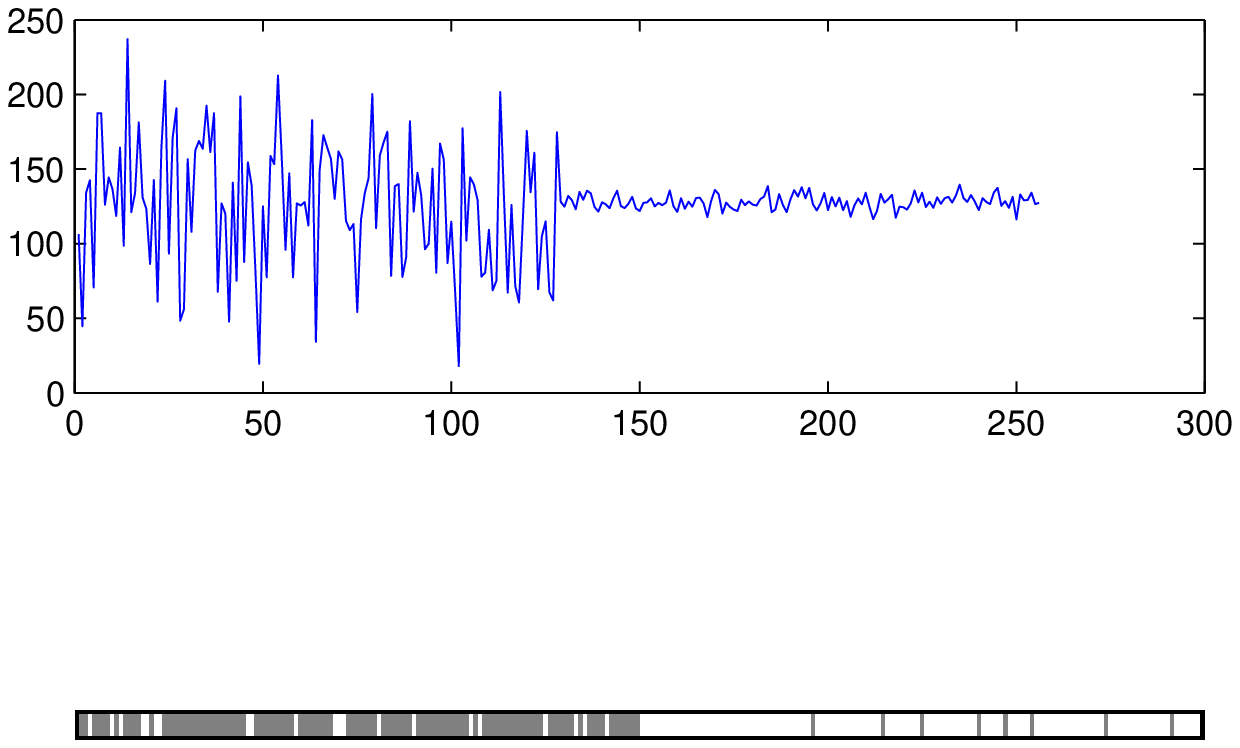}
 \caption{One dimensional, two classes. }
 \label{case_two}
\end{figure}

\begin{figure}[h]
 \centering
 \includegraphics[width=3in]{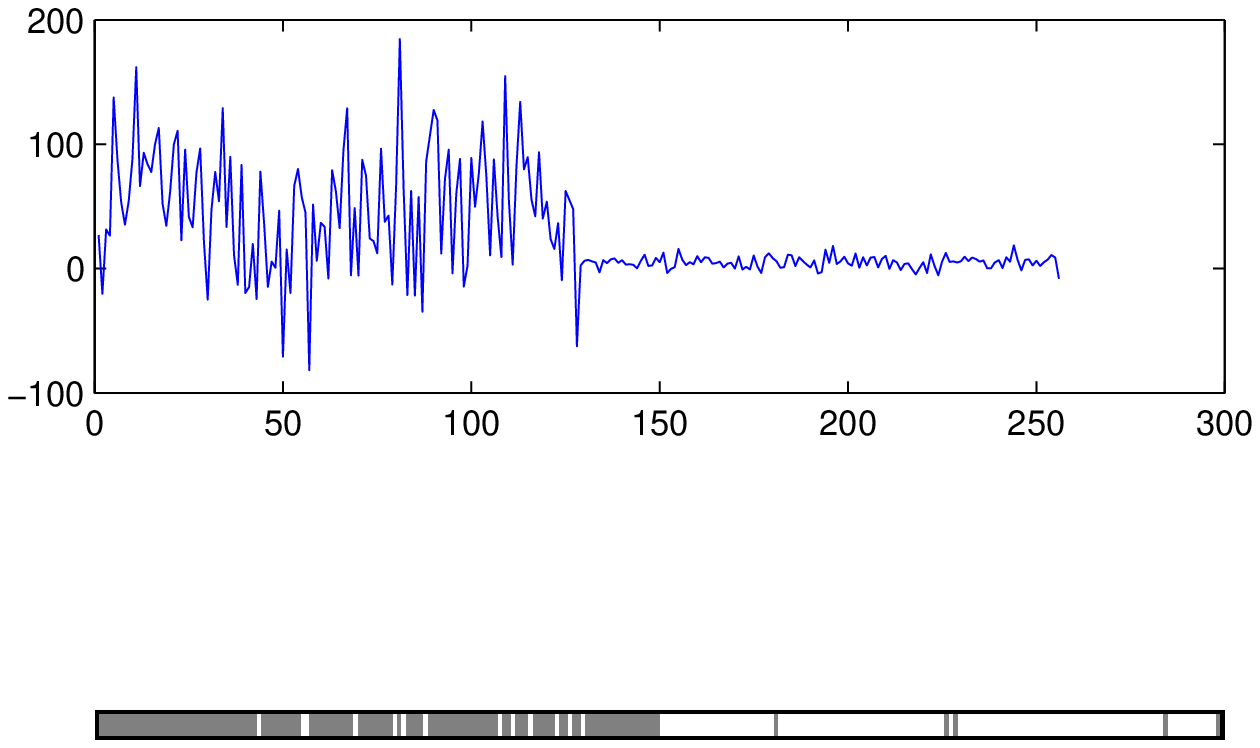}
 \caption{One dimensional, two classes. }
 \label{case_three}
\end{figure}

In Fig. \ref{case_two_dim}, we depict the result of the proposed
algorithm for a two dimensional mixed signal of two classes, with
one class having mean 200 and variance 400, and the other class
having mean 5 and variance 400. The signal is shown in the left part
of the figure. The classification result is shown in the right part
of the figure. The false classification ratios are $1.93$\% and
$0.52$\%. The size of the image is $32$ by $32$ pixels.

\begin{figure}[h]
 \centering
 \includegraphics[width=3in]{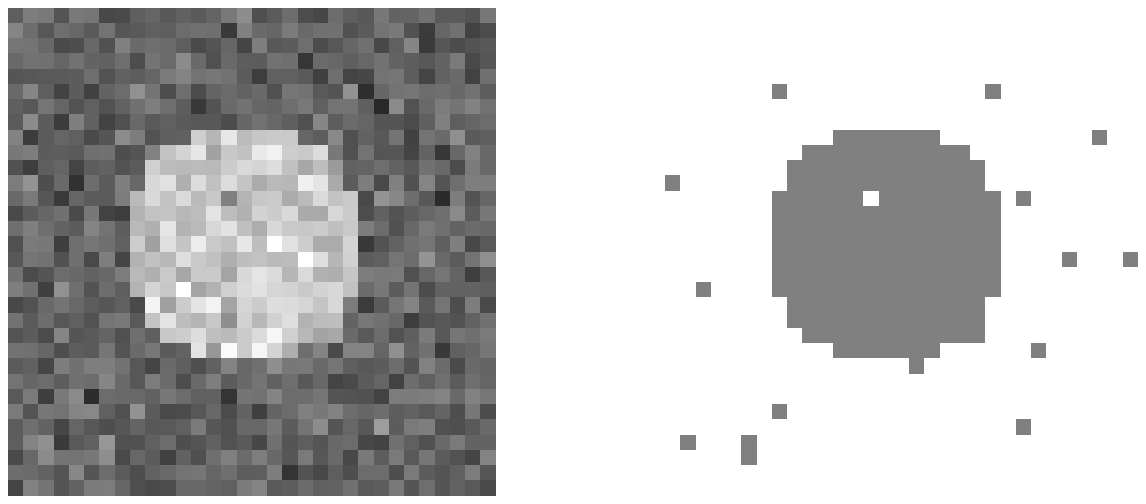}
 \caption{Two dimensional, two classes.}
 \label{case_two_dim}
\end{figure}

In summary, we find that the proposed classification algorithm is
effective and robust. The algorithm has low computational
complexity.

\section{Conclusion}

\label{sec_conclusion}

This paper proposes a blind classification algorithm for
non-stationary signals, which can be modeled as mixtures of signals
from several information sources. The proposed  algorithm is based
on data compression principles and relaxed continuous optimizations.
We present theoretical discussions, which show that our algorithm is
asymptotically optimal. Numerical results show that the proposed
algorithm is effective, robust and has low computational complexity.
The proposed algorithm  can be used to solve various multimedia
signal segmentation, analysis, and pattern recognition problems.

%


\bibliographystyle{latex8}
\bibliography{the_bib}

\end{document}